\documentclass[aps,pra, 11pt]{revtex4-2}
\usepackage{amssymb,amsmath,amsthm, verbatim, mathrsfs}

\usepackage{xcolor}

\usepackage{hyperref}
\hypersetup{
    colorlinks,
    linkcolor={red!50!black},
    citecolor={blue!50!black},
    urlcolor={blue!80!black}
}

\newcommand{\ket}[1]{| #1\rangle}
\newcommand{\bra}[1]{\langle #1 |}
\newcommand{\tr}{\mathrm{tr}\,}

\DeclareMathOperator{\GL}{GL}
\DeclareMathOperator{\SL}{SL}
\DeclareMathOperator{\PSL}{PSL}
\DeclareMathOperator{\U}{U}
\DeclareMathOperator{\SU}{SU}
\DeclareMathOperator{\SO}{SO}

\DeclareMathOperator{\End}{End}
\DeclareMathOperator{\Hom}{Hom}

\DeclareMathOperator{\Spec}{Spec}

\providecommand{\twomat}[4]{\left(\begin{matrix}#1&#2\\#3&#4\end{matrix}\right)}
\providecommand{\stwomat}[4]{\left(\begin{smallmatrix}#1&#2\\#3&#4\end{smallmatrix}\right)}

\theoremstyle{plain}
\newtheorem{thm}{Theorem}[section]

\theoremstyle{definition}

\begin{document}
\title{Effective Rationality for Local Unitary Invariants of Mixed States of Two Qubits}
\author{
Luca Candelori$^{1}$,~Vladimir Y. Chernyak$^{1,2}$,~John R. Klein$^{1}$,~Nick Rekuski$^{1}$
}

\affiliation{
$^{1}$\textit{Department of Mathematics, Wayne State University, Detroit, MI 48201, USA}\\
$^{2}$\textit{Department of Chemistry, Wayne State University, Detroit, MI 48201, USA}
}

\date{\today}
\begin{abstract}
We calculate the field of rational local unitary invariants for mixed states of two qubits, by employing methods from algebraic geometry. We prove that this field is rational (i.e. purely transcendental), and that it is generated by nine algebraically independent polynomial invariants. We do so by constructing a relative section, in the sense of invariant theory, whose Weyl group is a finite abelian group. From this construction, we are able to give explicit expressions for the generating invariants in terms of the Bloch matrix representation of mixed states of two qubits. We also prove similar rationality statements for the local unitary invariants of symmetric mixed states of two qubits. Our results apply to both complex-valued and real-valued invariants. 
\end{abstract}
\maketitle

\section{Introduction}

Quantum entanglement of two-level systems, or qubits, has long been recognized as an important resource in quantum computing and quantum information science, underpinning important concepts such as quantum teleportation and quantum key distribution, among many others (see \cite{Horodecki} for a survey). In most cases these applications require a classification of entanglement, whose goal is to identify all the possible levels of entanglement exhibited by multi-partite quantum states on a scale from no entanglement (i.e. separable states) to maximally entangled states (e.g. Bell  states). One possible approach to this problem is to classify the orbits of local unitary operators, where each orbit consists of states which are equivalent up to local unitary quantum evolution. Unfortunately, this simple approach quickly becomes intractable as soon as the number of parties increases to more than a few qubits. A more fruitul approach is to instead classify the functions on the Hilbert space of a multi-partite quantum system that are invariant under the action of local unitary operators. Restricting to the case of functions that are polynomial in the coordinates, a wide range of tools from computational invariant theory become available \cite{CIT} to solve this classification problem. For the case of pure states, this approach has led to a complete classification of polynomial local unitary invariants for up to four qubits \cite{Brylinski}, \cite{Wallach-Meyer}, \cite{Wallach-book}. In the case of mixed states, a complete classification has been given for the case of two mixed qubits in \cite{King}. It is also possible to obtain a complete classification of polynomial local unitary invariants for certain important qubit spaces, for example symmetric states \cite{Wallach:symmetric} and $X$-states \cite{X-states}. 

The classification of polynomial invariants, just like the classification of the orbits under local unitary action, becomes intractable very quickly: for pure states, as soon as the number of qubits is five or higher; for mixed states, the case of three qubits seems intractable using known methods. There are in principle algorithms for computing polynomial invariants that are guaranteed to terminate in a finite amount of time \cite{CIT}. However, even if one is willing to expend significant computational resources to calculate the next open case, the classification output is likely to be so intricate to decipher that in all probability it would be of little practical use. This is a well-known conundrum in invariant theory, evident already in the classical study of binary forms \cite{BinaryForms}. 

In this article, we propose a new approach to push the entanglement classification boundary a bit further. We propose to study {\em rational} local unitary invariants, as opposed to {\em polynomial} invariants. There are several advantages to this approach. First, well-developed methods from algebraic geometry become available \cite{PV}, \cite{CT-S} in the rational case that invite a study of the classification problem in a conceptual, qualitative manner, rather than by employing the raw, computational, purely algebraic approach of the polynomial case. Second, it turns out in practical applications that the structure of the field of rational invariants tends to be much simpler than that of the ring of polynomial invariants. In fact, it has been observed in many instances that the field of invariants tends to be {\em rational}, that is, {\em purely transcendental} \cite{dolgachev}, \cite{CT-S}. This means that the field is generated by a finite number of algebraically independent invariants (i.e. there are no polynomial relations between the generators), and the number of generators can moreover be calculated  beforehand by a simple dimension count. By contrast, the minimal number of generators for the ring of polynomial invariants cannot be easily ascertained a priori, and these generators tend to satisfy an intricate web of polynomial relations of very high degree. In terms of applications, it does not seem that passing to rational invariants presents a problem: in fact, many well-known entanglement invariants are not even algebraic, let alone polynomial. Such collection of invariants includes von Neumann entropy, concurrence, as well as the Horodecki invariant detecting violation of classical inequalities in mixed states \cite{Horodecki}.

We apply this new approach to the case of mixed states of two qubits. We show (Thms. \ref{thm:rationalityGeneralCase} and \ref{thm:generationGeneralCase}) that the field of rational local unitary invariants in this case is indeed rational, generated by nine invariants. We do so by employing the method of relative sections (i.e. the slice method) \cite{PV},\cite{CT-S}. This method consists in finding a suitable intersection of hyperplanes $S$ (the {\em relative section}, or {\em slice}) within the space of mixed sates so that the classification of rational invariants reduces to that of a finite  group $W(S)$, the {\em Weyl group} of $S$. In our case, we exhibit a relative section $S$ whose Weyl group is not only finite, but abelian. By a well-known 1913 result of E. Fischer  (see e.g. \cite[Prop. 4.3]{CT-S}), this immediately implies the rationality of the field of invariants. In addition, we are able to make our rationality result effective, by exhibiting an explicit set of nine generating invariants, constructed from our explicit knowledge of $S$. It is interesting to observe that the generators are in fact polynomial in the coordinates of the density matrix of a mixed state. More precisely, we express our invariants as polynomials in the entries of the Bloch matrix representation of a mixed state \cite{Gamel}, that is, in terms of 1- and 2-point correlation functions. Our methods allow to get a similar classification `for free' also in the case of symmetric states (Thm. \ref{thm:symmetricGeneration}), which in fact we present first, for ease of exposition. 

In order to apply methods of algebraic geometry, it simpler initially to complexify the space of mixed states and the group of local unitary operators, so as to obtain complex-valued invariants. In this way the ground field is algebraically closed, and the methods of classical algebraic geometry are readily available \cite{PV}. In applications, and especially if the goal is to implement entanglement invariants in terms of experimental measurements, it is important however to restrict to real-valued invariants. We do so in Section \ref{sec:realValued}, where our effective rationality result for mixed states of two qubits (symmetric or not) is extended to the real case (Thm. \ref{thm:generationRealInvariants}). Passing to the real case, where the ground field is no longer algebraically closed, requires slightly more sophisticated technical tools from modern algebraic geometry \cite{CT-S}, \citep{EGA1}, but the core ideas are the same. 

This new approach to the entanglement classification problem in terms of rational invariants can likely be applied to the case of multiple mixed qubits. The authors are currently working on extending the results of this article to this more general case. It might also be possible to extend our results to the case of qu-dits, that is, multi-level multi-partite quantum systems. In general, the authors hope to show in this article, above all, that the investigation of rational local unitary invariants can lead to fruitful progress in the understanding of entanglment of quantum states.

\section{Local Unitary Invariants of Mixed Qubit States}

We begin by recalling the basic setup of mixed qubit states and their local unitary invariants, as can be found in any basic textbook in quantum information science (e.g. \cite{NielsenChuang}). By a {\em qubit} we mean a state $\ket{\psi}$ in a two-dimensional Hilbert space.   More generally, a $n$-qubit (or $n$-qubits) is a state in a tensor product  $\bigotimes^n \mathbb{C}^2$ of two-dimensional Hilbert spaces. A {\em mixed state of $n$-qubits} is a density matrix of the form 
$$
\rho = \sum p_i \ket{\psi_i}\bra{\psi_i}  
$$
where $p_i \geq 0$ are a finite set of real numbers satisfying $\sum p_i = 1$ and  $\{\ket{\psi_i}\}$ is a finite set of $n$-qubits. The set of all such mixed states coincides with the set of Hermitian, non-negative, linear operators on $\bigotimes^n \mathbb{C}^2$ of trace one. In this article, we restrict our attention to the case of mixed states of $n=2$ qubits. 

The quantum evolution of a mixed state $\rho$ is given by conjugation by a unitary matrix, $\rho \mapsto U\rho U^{\dagger}$. In the case of two mixed qubits, we have a special class of unitary operators of the form $U_1\otimes U_2$, where each $U_i$ acts on a separate qubit. We call these {\em local unitary operators}. 

We are interested in calculating algebraic functions $f$ (real- or complex-valued) on mixed states of two qubits that are invariant under the action of local unitary operators, that is, 
\begin{equation}
\label{eqn:invariantAction}
f(U\rho U^{\dagger}) = f(\rho)
\end{equation}
for every unitary $U:\mathbb{C}^2\otimes\mathbb{C}^2 \rightarrow \mathbb{C}^2\otimes\mathbb{C}^2 $ of the form $U = U_1\otimes U_2$. In particular:

\begin{itemize}
\item[(1)] When $f$ is polynomial in the entries of the matrix $\rho$, we say that $f$ is a {\em polynomial local unitary invariant}.
\item[(2)] When $f = f_1/f_2$ is a rational function in the entries of the matrix $\rho$, we say that $f$ is a {\em rational local unitary invariant}.

\end{itemize} 

The set of all polynomial (or rational) local unitary invariants on mixed states of two qubits does not come equipped with a nice algebraic structure, since the space of density matrices is not an algebraic variety (it is however a {\em semi}-algebraic variety \cite{Gamel}). However, if we remove the positivity condition on a density matrix, and instead work with the Liouville space 
$$
\mathscr{L} = \{ \rho : \rho^{\dagger} = \rho, \tr \rho = 1\},
$$ 
we obtain an affine space $\mathscr{L} \simeq \mathbb{A}_{\mathbb{R}}^{15}$ which is, in particular, an affine algebraic variety over the real numbers. The set of (real-valued) polynomial local unitary invariants on $\mathscr{L}$ now forms a ring, denoted by $\mathbb{R}[\mathscr{L}]^G$, where $G = \U(2)\times \U(2)$ is the group of local unitary operators. This ring can now be investigated using the tools and algorithms of classical invariant theory \cite{CIT}. For example, a full presentation for this ring can be found in \cite{King}, where the authors give 25 generators and 63 relations. Similarly, the set of (real-valued) rational local unitary invariants  on $\mathscr{L}$ forms a field, denoted by $\mathbb{R}(\mathscr{L})^G$, whose structure we intend to investigate in this article. In contrast to the polynomial case, we will show that this field has a much simpler structure, consisting of nine invariants an no relations.  

It is easier to first calculate complex-valued invariants, and then afterwards pass to the real case. Therefore we let 
$$
\mathscr{L}_{\mathbb{C}} = \{ \rho : \tr \rho = 1\} \simeq \mathbb{A}^{15}_{\mathbb{C}}
$$ 
be the complexification of the Liouville space $\mathscr{L}$. The complexification of $\U(2)\times \U(2)$, which is given by $\GL_2(\mathbb{C})\times \GL_2(\mathbb{C})$, acts on $\mathscr{L}_{\mathbb{C}}$ by conjugation. In fact, because of \eqref{eqn:invariantAction}, for the calculation of invariants it suffices to consider the action of the subgroup $\SU(2)\times \SU(2)$, whose complexification is $\SL_2(\mathbb{C})\times \SL_2(\mathbb{C})$. We are therefore reduced to the calculation of 
$$
\mathbb{R}(\mathscr{L})^{\SU(2)\times \SU(2)}\otimes \mathbb{C} = \mathbb{C}(\mathscr{L}_{\mathbb{C}})^{\SL_2(\mathbb{C})\times \SL_2(\mathbb{C})},
$$
as follows from a simple application of base-change properties of fields of invariants.

\section{Lie algebra $\mathfrak{sl_2}$ and the adjoint representation}

We shall require some basic facts about  $\mathrm{SL}_2(\mathbb{C})$ and its Lie algebra. Recall that $\mathrm{SL}_2(\mathbb{C})$ is a 3-dimensional complex Lie group, whose Lie algebra is 
$$
\mathfrak{sl}_2 = \{ M \in \mathrm{Mat}_2(\mathbb{C}) : \tr M = 0 \}
$$
with Lie bracket given by the standard commutator bracket $[A, B] = AB - BA$, and Killing form  
$$
(A, B) = 4\,\tr AB.
$$

For ease of notation, let $V = \mathfrak{sl}_2$. Equipped with the symmetric bilinear form $(,)$, the vector space $V$ is an orthogonal space with structure group given by the orthogonal group 
$$
\mathrm{O}(V) = \{ g \in \mathrm{GL}(V) : (gu, gv) = (u,v)\; \forall u,v \in V \} \simeq \mathrm{O}_3(\mathbb{C}).
$$
We also let
$$
\mathrm{SO}(V) = \{ g \in \mathrm{O}(V) : \det g = 1 \} \simeq \mathrm{SO}_3(\mathbb{C})
$$
be its corresponding {\em special} orthogonal group of orientation-preserving transformations.

There is an action of $\mathrm{SL}_2(\mathbb{C})$ on $V$ (the {\em adjoint} action) given by matrix conjugation, which factors through $\mathrm{PSL}_2(\mathbb{C}) = \mathrm{SL}_2(\mathbb{C})/{\pm I}$. This is a three-dimensional  irreducible algebraic representation of $\mathrm{SL}_2(\mathbb{C})$, the spin-1 representation. The Killing form is preserved by this action, and we obtain a well-known isomorphism 
\begin{equation}
\label{eq:PSL2-SO3-iso}
\PSL_2(\mathbb{C}) \simeq \mathrm{SO}(V) \simeq \mathrm{SO}_3(\mathbb{C})
\end{equation} 
under which the adjoint representation $V$ of $\mathrm{SL}_2(\mathbb{C})$  goes over to the standard representation of $\mathrm{SO}_3(\mathbb{C})$. Note that under this isomorphism, for any two vectors $u,v \in V$, 
\begin{align*}
[u,v] &= u\times v \\
(u,v) &= u\cdot v
\end{align*}
that is, the Lie bracket corresponds to the cross-product of two vectors in $V \simeq \mathbb{C}^3$, while the Killing from (by construction) corresponds to the standard dot product $u\cdot v = u_1v_1 + u_2v_2 + u_3v_3$.  We also let 
$$
\lVert v \rVert = \sqrt{(v,v)} 
$$
be the vector norm of a vector $v\in V$. Note that the dot-product is {\em not} an inner product: the norm is in general complex-valued, and in particular it is possible to have non-zero vectors $v\in V$ such that $\lVert v \rVert = 0$. 

The cross-product and dot-product satisfy the following identities, that we record for convenience:
\begin{align}
\label{eqn:identity1}
a\times (b\times c) &=  (a\cdot c)b - (a\cdot b)c \\
\label{eqn:identity2}
(a\times b)\cdot (c\times d) &= (a\cdot c)(b\cdot d)- (b\cdot c)(a\cdot d)
\end{align}
for all $a,b,c,d \in V$. These identities are well-known that are easy to verify directly. Equivalent identities can be written in terms of the Lie bracket and the Killing form on $\mathfrak{sl}_2$.

\section{The Bloch matrix representation}

We now describe a representation of Liouville space that is particularly convenient for the study of local unitary invariants. This representation generalizes the well-known Bloch ball representation for a mixed state of one qubit, and is therefore known as the {\em Bloch representation}. A detailed study of this representation is given in \cite{Gamel}. To define it, consider the Pauli matrices 
\begin{equation}
\label{eqn:PauliMatrices}
\sigma_1 = \twomat 0110, \quad \sigma_2 = \twomat 0{-i}i0, \quad \sigma_3 = \twomat {1}00{-1}
\end{equation}
and let $\sigma_0 = \stwomat 1001$ be the identity matrix. Any mixed state of one qubit can be expressed as a linear combination 
$$
\rho = \frac{1}{2}\left(\sigma_0 + v_1\sigma_1 + v_2\sigma_2 + v_3\sigma_3\right)
$$
with coefficients $v = (v_1, v_2, v_3)^t \in \mathbb{R}^3$ given by $v_i = \tr(\rho \sigma_i)$. The positivity condition on $\rho$ ensures that $v$ lies within a ball of radius one centered at the origin, known as the {\em Bloch ball}. In this representation, quantum evolution on $\rho$ by operators in $\SU(2)$ correspond to rotations on $v$ by elements of $\SO_3(\mathbb{R})$.
Complexifying, the Pauli matrices form a basis for the Lie algebra $V = \mathfrak{sl}_2$ that is orthogonal with respect to the Killing form $(,)$. Removing the positivity condition, we see that the map  $\rho \mapsto v$ is an isomorphism of vector spaces from the Liouville space of one mixed qubit to the adjoint representation $V$ of $\SL_2(\mathbb{C})$. Conjugation by $\SL_2(\mathbb{C})$ on an element $\rho$ of Liouville space corresponds to the standard action of $\SO(V) \simeq \SO_3(\mathbb{C})$ on the vector $v\in V$. The isomorphism $\rho \mapsto v$ is thus precisely the map underlying the isomorphism of groups given by \eqref{eq:PSL2-SO3-iso}.  

For the case of two qubits, there are two copies $V_1, V_2$ of the adjoint representation $V = \mathfrak{sl}_2$ of $\SL_2(\mathbb{C})$ and a basis for the space of operators on the tensor product $V_1\otimes V_2$ can be given by the Dirac matrices $\sigma_i\otimes \sigma_j$, $i,j = 0,1,2,3$. We may then express $\rho \in \mathscr{L}_{\mathbb{C}}$ as a linear combination
$$
\rho = \frac{1}{4}\left( \sum_{i,j = 0}^3 c_{ij} \, \sigma_i\otimes \sigma_j    \right),
$$
with coefficients $c_{ij} \in \mathbb{C}$. It is useful to organize these 16 coefficients in terms of {\em correlation functions}:

\begin{itemize}
\item[(i)] The 0-point correlation function: $$c_{00} = \tr(\rho) = 1$$
\item[(ii)]The 1-point correlation functions:
\begin{align*}
u_{1,i} &= c_{i0} = \tr(\rho\, \sigma_i\otimes \sigma_0) \\
u_{2,j} &= c_{0j} = \tr(\rho\, \sigma_0\otimes \sigma_j),
\end{align*} 
with $i,j = 1,2,3$.
\item[(iii)] The 2-point correlation functions:
$$
c_{ij} = \tr(\rho\, \sigma_i\otimes \sigma_j)
$$
with $i,j = 1,2,3$.
\end{itemize}
We then organize the correlation functions into the {\em Bloch matrix}:
\begin{equation}
\label{eqn:blochRep}
B(\rho) = \left( \begin{tabular}{c|c}
 1 & $u_2^t$ \\
 \hline
 $u_1$ & $C$ \\ 
 \end{tabular} \right) \in \mathbb{C}^{15},
\end{equation}
where $u_1 = (u_{1,i}) \in V_1, u_2 = (u_{2,j}) \in V_2$ are vectors, each belonging to a copy of $V = \mathfrak{sl}_2$, and $C = (c_{ij}) \in \Hom(V_2, V_1)$ is the matrix of a linear tranformation $V_2\rightarrow V_1$. In this representation, the action of $\SL_2(\mathbb{C})\times \SL_2(\mathbb{C})$ (the complexification of the group of local unitary operators) on the density matrix $\rho \in \mathscr{L}_{\mathbb{C}}$ goes over to an action of the special orthogonal operators $(g_1, g_2) \in \SO(V_1)\times \SO(V_2) \simeq \SO_3(\mathbb{C})\times \SO_3(\mathbb{C})$ given by 
\begin{equation}
\label{eqn:BlochAction}
(g_1,g_2)B(\rho) = \left( \begin{tabular}{c|c}
 1 & $u_2^tg_2^t$ \\
 \hline
 $g_1u_1$ & $g_1Cg_2^t$ \\ 
 \end{tabular} \right).
\end{equation}
Therefore, we get a decomposition 
$$
\mathscr{L}_{\mathbb{C}} \simeq V\boxtimes 1\oplus 1\boxtimes V \oplus V\boxtimes  V
$$
of complexified Liouville space into irreducible representations of $\SO_3(\mathbb{C})\times \SO_3(\mathbb{C})$. Here `$1$' denotes the trivial representation of $\SO_3(\mathbb{C})$, and for two representations $V_1, V_2$ of a group $G$, we have used the external tensor product notation $\boxtimes$ to denote the tensor product $V_1\otimes V_2$ as a representation of $G\times G$, where each copy of $G$ acts on each tensor factor. This is to emphasize that $V\boxtimes  V \simeq V_1\otimes V_2$ is indeed irreducible as a representation of $\SO_3(\mathbb{C})\times \SO_3(\mathbb{C})$. Moreover, we have used the fact that $V$ is canonically isomorphic to its dual representation (via the Killing form) so that there is a canonical isomorphism $\End(V) = V\otimes V$.

\section{Relative sections and rationality}
\label{sec:relativeSections}

We now turn to the calculation of the field of invariants
$$
\mathbb{C}(\mathscr{L}_{\mathbb{C}})^{\SL_2(\mathbb{C})\times \SL_2(\mathbb{C})} = \mathbb{C}(\mathscr{L}_{\mathbb{C}})^{\SO_3(\mathbb{C})\times \SO_3(\mathbb{C})},
$$
where the action of the group $\SO_3(\mathbb{C})\times \SO_3(\mathbb{C})$ on complexified Liouville space is given by \eqref{eqn:BlochAction}. For this calculation we will use the method of {\em relative sections} \cite[Ch.2]{PV}, also known as the {\em slice method}. This method applies to the general situation of $G$ a reductive algebraic group acting on an irreducible algebraic variety $X$ over $\mathbb{C}$. In this case, a subvariety $S\subseteq X$ is a {\em relative section} if there exists a $G$-invariant (dense) Zariski-open subset $X_0 \subseteq X$ satisfying the following properties:
    \begin{itemize}
      \item[(i)]{
        The Zariski closure of the orbit space $GS$ is all of $X$ (that is, $\overline{GS} = X$).
      }
      \item[(ii)]{
        Let $N=N(S)=\{g\in G:gS \subseteq S\}$ be the normalizer of $S$, and let $S_0=X_0\cap S$.
        Then $gS_0\cap S_0\neq\emptyset$ implies that $g\in N$.
      }
    \end{itemize} 

Relative sections are important from both a theoretical and computational point of view, since they reduce the calculations of the $G$-invariants to those of the {\em Weyl group} of $S$, a smaller group. To define it, note that the action of the normalizer $N$ on $S$ may not be faithful, that is, there could be elements $g \in N$ that fix every point $x\in S$. It is therefore more convenient to quotient this group by the {\em centralizer} $Z(S) = \{g\in N:gx = x ,\,\forall x \in S\}$. The quotient
$$
W(S) := N(S)/Z(S)
$$     
acts faithfully on $S$, and it is called the {\em Weyl group} of $S$. This terminology originates from the special case of the adjoint representation of a complex semi-simple connected Lie group, where the Cartan subalgebra is a relative section, and its Weyl group corresponds to the Weyl group of the Lie group itself.

Given any relative section $S\subseteq X$ with Weyl group $W = W(S)$, we may take a rational invariant function $f \in  \mathbb{C}(X)^G$ and restrict its domain to $S$, thus obtaining an invariant in  $\mathbb{C}(S)^W$. It can be shown \cite[Ch. 2.5]{PV}  that this restriction of invariants induces an isomorphism 
\begin{equation}
\label{eqn:relSecIso}
\mathbb{C}(X)^G \stackrel{\simeq}\longrightarrow \mathbb{C}(S)^W
\end{equation}
of invariant fields. In particular, if $S$ can be chosen so that $W = W(S)$ is finite (or equivalently, so that $\dim S = \dim X - \dim G$) then the calculation of the field of invariant reduces to the case of a finite group action (of smaller dimension), for which there are much more efficient algorithms \cite{CIT} than for an arbitrary reductive group.

Even when a convenient relative section can be found, there is no a priori reason to expect the field of invariants $\mathbb{C}(X)^G$ to have a simple structure. In fact, for the group $G = \SO_3(\mathbb{C})$ (and products of $G$) and for finite groups, the field of invariants corresponds to the fraction field of the ring of polynomial invariants  $\mathbb{C}[X]^G$ \cite[Lemmas 2.1, 2.2]{CT-S}, which, as we saw earlier, can have an extremely rich and complicated structure. It may therefore come as a surprise that, when computed in practice, the field of invariants tends to be {\em rational}, that is, {\em purely transcendental} \cite{dolgachev}, \cite{CT-S}. This means that $\mathbb{C}(X)^G$ can be generated by $n$ elements $f_1, \ldots, f_n$ (where $n = \dim X - \dim G$) that are {\em algebraically independent}, that is, there are no polynomial relations between the $f_i$'s with coefficients in $\mathbb{C}$. We will show that this is indeed the case for the field $\mathbb{C}(\mathscr{L}_{\mathbb{C}})^{\SO_3(\mathbb{C})\times \SO_3(\mathbb{C})}$.

\section{Symmetric states}
\label{sec:symStates}

We first show rationality and calculate generators for the field of invariants of {\em symmetric} mixed states of two qubits. This is the  subspace $\mathscr{L}_{\rm{sym}} \subseteq \mathscr{L}$ of the Liouville space for mixed states of two qubits consisting of states that are left unchanged by the map $v_1\otimes v_2 \mapsto v_2\otimes v_1$ swapping the ordering of tensors. Equivalently, symmetric states are characterized in terms of the Bloch representation \eqref{eqn:blochRep} by the property that $u = u_1 = u_2$ and $C = C^t$, that is, the matrix of 2-point correlation functions is symmetric.  In particular, for symmetric states the two copies $V_1, V_2$ of the adjoint representation $V = \mathfrak{sl}_2$ corresponding to the two different qubits are identified, and the symmetric matrix $C$ is identified with a linear endomorphism $V\rightarrow V$. 

The full group of local unitary operators $\SU(2)\times \SU(2)$ does not preserve the subspace of symmetric states $\mathscr{L}_{\rm{sym}}$. Instead, we have an action of $\SU(2)$ embedded inside $\SU(2)\times \SU(2)$ via the diagonal map $g \mapsto (g,g)$. Complexifying, the corresponding $\SO_3(\mathbb{C})$-action on the Bloch representation \eqref{eqn:blochRep} is given by 
$$
gB(\rho) = \left( \begin{tabular}{c|c}
 1 & $u^tg^t$ \\
 \hline
 $gu$ & $gCg^t$ \\ 
 \end{tabular} \right),
$$
where $C = C^t$ is symmetric. Symmetric states are therefore isomorphic to $V\otimes \mathrm{Sym}^2(V)$, as a representation of $\PSL_2(\mathbb{C}) \simeq \SO_3(\mathbb{C})$.

We now construct a relative section (in the sense of Section \ref{sec:relativeSections}) for the action of $\mathrm{SO}_3(\mathbb{C})$ on (complexified) symmetric states $\mathscr{L}_{\rm{sym},\mathbb{C}}$. For ease of notation, let 
$$X = \mathscr{L}_{\mathrm{sym},\mathbb{C}}, \quad  G = \mathrm{SO}_3(\mathbb{C}). $$  

Let $S\subseteq X$ be the linear subvariety defined by 
\begin{equation}
\label{equation:symmetricSection}
S := \left\{ \left( \begin{tabular}{c|c}
 $1$ & $u^t$ \\
 \hline
 $u$ & $C$  \\ 
 \end{tabular}\right):  u_2 = u_3 = 0, c_{12} = c_{21} = 0   \right\}, 
\end{equation}
where $(u_i)$ are the coordinates of the vector $u \in V$ and $(c_{ij})$ are the matrix entries of $C \in \End(V)$ with respect to the standard basis $\{e_1, e_2, e_3\}$ of $V \simeq \mathbb{C}^3$. We first show that the orbit space $GS$ is Zariski-dense in  $X$. Given an arbitrary element $x \in X$, represented in Bloch matrix form by $(u,C)$, let 
\begin{equation}
\label{eqn:uv-sym}
v := u\times Cu, \quad w := u\times v.
\end{equation}
Let $X_0 \subseteq X$ be the $G$-invariant Zariski-open subset defined by  $u\cdot u \neq 0$ and $v\cdot v \neq 0$. Note that by  \eqref{eqn:identity2} with $a = b = u$ and $c = d = v$ we have that $w\cdot w = (u\cdot u)(v\cdot v)$, so that we also have $w\cdot w \neq 0$. Because all the vector norms are non-zero, the set $\vec{u} = u/\lVert u \rVert, \vec{v} = v/\lVert v \rVert, \vec{w} = w/\lVert w \rVert$ is an orthonormal basis for $V$. Therefore, there exists an orthogonal transformation $g \in \mathrm{O}(V)$ such that $e_1 = g\vec{u}, e_2 = g\vec{v}, e_3 = g\vec{w}$. Moreover, 
\begin{align*}
(e_1\times e_2)\cdot e_3 &= (g\vec{u} \times g\vec{v})\cdot g\vec{w} \\
& = g(\vec{u}\times \vec{v})  \cdot g\vec{w} \\
& = (\vec{u}\times \vec{v}) \cdot \vec{w}
\end{align*}
so that the orientation of the bases are preserved by $g$, and therefore $g \in \mathrm{SO}(V)$.  Applying $g$ to the element $(u,C)$ we get: 
\begin{align}
\label{eqn:section}
g C g^t &= \left(\begin{array}{ccc}
\vec{u}\cdot C\vec{u} & 0 & \vec{u}\cdot C\vec{w} \\ 
0 & \vec{v}\cdot C\vec{v} & \vec{v}\cdot C\vec{w} \\ 
\vec{u}\cdot C\vec{w} & \vec{v}\cdot C\vec{w} & \vec{w}\cdot C\vec{w}
\end{array} \right), \\
\quad gu &= \left(\begin{array}{c}
\lVert u \rVert \\ 
0 \\ 
0
\end{array} \right)
\end{align} 
since $\vec{u}\cdot C\vec{v} =  \vec{v}\cdot C\vec{u} = 0$, by construction. Therefore the element $(u,C) \in X_0$ belongs to the $G$-orbit of $S$, so that  $\overline{GS} = V$. 

Next, we compute the Weyl group $W(S)$ of $S$. Suppose $(u,C) \in S$ is in the form \eqref{equation:symmetricSection}. In order for $g\in \mathrm{SO}(V)$ to belong to the normalizer $N(S)$, it must first preserve the subspace spanned by  $e_1$, so that $$g = \left( \begin{tabular}{c|c}
 $\pm 1$ & $0$ \\
 \hline
 $0$ & $g_0$  \\ 
 \end{tabular}\right)
$$
for some orthogonal 2-by-2 matrix $g_0$. Writing $C = \left( \begin{tabular}{c|c}
 $a$ & $b^t$ \\
 \hline
 $b$ & $C_0$  \\ 
 \end{tabular}\right)
$ in block form, where $b^t = (0, b_0)$, we deduce that $g$ belongs to $N(S)$ only if $g_0$ preserves the subspace spanned by $b$, so that 
$$
g_0 =  \left(\begin{array}{cc}
\pm 1 & 0 \\ 
0 & \pm 1
\end{array} \right).
$$
Therefore, $N(S)$ is isomorphic to the subgroup of $\mathrm{SO}_3(\mathbb{C})$ consisting of diagonal matrices. The centralizer $Z(S)$ in this case is trivial, so that the Weyl group $W(S)$ equals $N(S)$  and it is given by
$$
W(S) \simeq \left\{ g = \left(\begin{array}{ccc}
\pm 1 & 0 & 0 \\ 
0 & \pm 1 & 0 \\ 
0 & 0 & \pm 1
\end{array} \right): \det g =1     \right\}  \subseteq \mathrm{SO}_3(\mathbb{C}).
$$
This is an abelian group, abstractly isomorphic to $\mathbb{Z}/2\mathbb{Z} \times \mathbb{Z}/2\mathbb{Z}$, the Klein 4-group. For example, it can be generated by the two order 2 matrices
$$
\left(\begin{array}{ccc}
 1 & 0 & 0 \\ 
0 & - 1 & 0 \\ 
0 & 0 & - 1
\end{array} \right), \quad \left(\begin{array}{ccc}
 -1 & 0 & 0 \\ 
0 &  1 & 0 \\ 
0 & 0 & - 1
\end{array} \right).
$$
By the same calculation, we can see that $S$ is in fact a relative section: if $x \in S_0 = X_0\cap X$ and $g\in G$ satisfies $gx \in S_0$, then we must have that $g \in N$. 

The existence of the relative section $S$ immediately implies the rationality of the field of invariants $\mathbb{C}(X)^G$. In particular, since $S$ is a linear representation of a finite abelian group $W$, the field of invariants $\mathbb{C}(S)^W$ is rational, by a well-known 1913 result of E. Fischer (see e.g. \cite[Prop. 4.3]{CT-S}). This field is thus generated by $\dim S - \dim W = 6$ algebraically independent invariants $s_1, \ldots, s_6$. By the restriction isomorphism \eqref{eqn:relSecIso} we deduce that the field $\mathbb{C}(X)^G$ is also rational, generated by six invariants $f_1, \ldots f_6$, whose restriction to $S$ are the six invariants $s_1, \ldots, s_6$.

\section{Orbit separation }

Having shown that the field of invariants $\mathbb{C}(X)^G$ (with $X = \mathscr{L}_{\rm{sym},\mathbb{C}}$ and $G = \SO_3(\mathbb{C})$) is rational, we now turn to the problem of effectively calculating generators. For this problem, we use the `separation of orbits' criterion \cite[Ch. 2.3]{PV}. To explain this geometric criterion, note that when two points $x_1, x_2$ are in the same $G$-orbit, then clearly $f(x_1) = f(x_2)$ for every invariant function, but the converse is not necessarily true. In particular, we say that a set $\{f_1, \ldots, f_n\} \subseteq \mathbb{C}(X)^G$ {\em separates orbits in general position} if there exists a Zariski-open (dense) subset $U\subseteq X$ such that, for every pair of elements $x_1, x_2 \in U$, if $f_i(x_1) = f_i(x_2)$ for all $i=1, \ldots, n$ then $x_1$ and $x_2$ belong to the same $G$-orbit, that is, there exists $g\in G$ such that $gx_1 = x_2$. We then have the following criterion for generation of the field of invariants:

\begin{thm}[\cite{PV}, Lemma 2.1]
\label{thm:separationTest}
Suppose $\{f_1, \ldots, f_n\} \subseteq \mathbb{C}(X)^G$ separates orbits in general position. Then $f_1, \ldots, f_n$ generate the field $\mathbb{C}(X)^G$. 
\end{thm} 

In the case of symmetric states, whose Bloch representation recall consists of a vector $u\in V \simeq \mathbb{C}^3$ and a symmetric matrix $C \in \End(V)$, let 
$$
v = u\times Cu, \quad w = u\times v
$$
as in \eqref{eqn:uv-sym}, and let 

\begin{align*}
f_1 = u\cdot u, &\quad f_2 = v\cdot v \\
f_3 = u\cdot Cu, &\quad f_4 = v\cdot  Cv \\
f_5 = w\cdot Cw, &\quad f_6 = w\cdot Cv.  
\end{align*}

We can then show that the set $\{f_1, \ldots, f_6\}$ generates the field of invariants:

\begin{thm}
\label{thm:symmetricGeneration}
The functions $f_1, \ldots f_6 \in \mathbb{C}(\mathscr{L}_{\rm{sym}})$ are algebraically independent $\SO_3(\mathbb{C})$-invariant functions, and they generate the field $\mathbb{C}(\mathscr{L}_{\rm{sym}})^{\SO_3(\mathbb{C})}$.
\end{thm}

\begin{proof}
We first show that these six functions are invariant. Note that the action of $g \in \SO_3(\mathbb{C})$ sends $u \mapsto gu$ and $C\mapsto gCg^t$, so that $v = u\times Cu$ is sent to $gv$ and $w = u\times v$ to $gw$. It follows easily that all the six functions $f_1, \ldots, f_6$ are invariant. For example, for $f_6$ we have 
$$
gw \cdot gCg^tgv = gw\cdot gCv = w\cdot Cv,
$$
and similarly for the other $f_i$'s. 

Next, we show that these six functions generate the whole field of rational invariants, by showing that they separate orbits in general position (Thm. \ref{thm:separationTest}). Let $X_0\subseteq X$ be the $G$-invariant, Zariski-open set given by $u\cdot u \neq 0$ and $v\cdot v \neq 0$. Suppose        
$$
x = \left( \begin{tabular}{c|c}
 1 & $u^t$ \\
 \hline
 $u$ & $C$ \\ 
 \end{tabular} \right),\quad 
 x' = \left( \begin{tabular}{c|c}
 1 & $(u')^t$ \\
 \hline
 $u'$ & $C'$ \\ 
 \end{tabular} \right),
$$
are two elements of $X_0$ satisfying $f_i(x) = f_i(x')$ for $i=1,\ldots, 6$. We need to show that $x,x'$ are in the same $\SO_3(\mathbb{C})$-orbit. By the same argument as in Section \ref{sec:symStates}, there exist an element $g\in \SO_3(\mathbb{C})$ so that $gx$ is in the form \eqref{eqn:section}. By identity \eqref{eqn:identity1} with $a=b=u$ and $c = Cu$, we may write 
$$
w = (u\cdot Cu)u - (u\cdot u)Cu
$$
so that 
$$
u\cdot Cw = Cu \cdot w = (u\cdot Cu)^2 - (u\cdot u)(Cu\cdot Cu).
$$
On the other hand, by identity \eqref{eqn:identity2} with $a=c=u$ and $b=d=Cu$ we have 
$$
v\cdot v = (u\times Cu) \cdot (u\times Cu) = (u\cdot u)(Cu\cdot Cu) - (u\cdot Cu)^2 
$$
so that
$$
u\cdot Cw = -(v\cdot v) = -f_2.
$$ 
We may then express $gx$ in terms of the invariant functions $f_1(x), \ldots, f_6(x)$ as follows:

\begin{align*}
g C g^t &= \left(\begin{array}{ccc}
f_3/f_1& 0 & -f_2^{1/2}/f_1 \\ 
0 & f_4/f_2 &  f_6/f_2f_1^{1/2} \\ 
-f_2^{1/2}/f_1 & f_6/f_2f_1^{1/2} & f_5/f_1f_2
\end{array} \right), \\
\quad gu &= \left(\begin{array}{c}
f_1^{1/2} \\ 
0 \\ 
0
\end{array} \right).
\end{align*} 
Similarly, we can find $g'\in \SO_3(\mathbb{C})$ so that $g'x'$ is an identical expression in terms of the functions $f_1(x'), \ldots, f_6(x')$. Because $f_i(x) = f_i(x')$ for all $i$, it follows that $gx = g'x'$ and therefore $x = g^{-1}g'x'$, so that $x,x'$ are in the same orbit. 

We have thus proved that $f_1, \ldots, f_6$ generate the field $\mathbb{C}(\mathscr{L}_{\rm{sym}})^{\SO_3(\mathbb{C})}$. We also know that this field is purely transcendental of transcendence degree six. Therefore $\{f_1, \ldots, f_6\}$ must be algebraically independent, otherwise they would generate a field of transcendence degree less than six.

\end{proof}

Note that the generators $f_1, \ldots, f_6$ are in fact polynomial functions, and not just rational functions. They also have particularly simple expressions in terms of the coordinates $u = (u_i)$ and $C = (c_{ij})$, which can be obtained by applying well-known formulas for the dot-product and the cross-product.

\section{General case}

We now apply the technique of relative sections to the case of the full Liouville space $\mathscr{L}_{\mathbb{C}}$ of mixed states of two qubits. We show that the field of invariants is rational, and provide a simple set of generating elements for this field. For ease of notation, let
$$
X = \mathscr{L}_{\mathbb{C}}, \quad G = \SO_3(\mathbb{C})\times \SO_3(\mathbb{C})
$$
where recall the action of $G$ on an element $x \in X$ in Bloch matrix form is given by \eqref{eqn:BlochAction}. Denote by $V_1, V_2$ the two copies of the adjoint representation $V = \mathfrak{sl}_2$ corresponding to each qubit, so that an element $x \in X$ is represented in Bloch matrix form by a triple $(u_1, u_2, C)$ where $u_i \in V_i$ and $C \in \Hom(V_2, V_1)$. Let 
\begin{equation}
\label{equation:generalSection}
S := \left\{ \left( \begin{tabular}{c|c}
 $1$ & $u_2^t$ \\
 \hline
 $u_1$ & $C$  \\ 
 \end{tabular}\right):  u_{1,2} = u_{1,3} = 0,u_{2,2} = u_{2,3} = 0, c_{12} = c_{21} = 0   \right\}, 
\end{equation}
where $u_1 = (u_{1,i}) \in V_1, u_2 = (u_{2,i}) \in V_1,  C = (c_{i,j}) \in \Hom(V_2,V_1)$ are the coordinates of each component of the Bloch matrix with respect to the standard basis $\{e_1, e_2, e_3\}$ of $\mathbb{C}^3$. We first show that this a relative section for the action of $G$ on $X$ (in the sense of Section \ref{sec:relativeSections}). Note that the Killing form on $V$ gives a canonical isomorphism between $V$ and its dual $V^*$, so that the matrix transpose $C^T$ can be viewed as a linear transformation $C^T \in \Hom(V_1,V_2)$. In this notation, let 
\begin{align*}
v_1 := u_1 \times Cu_2 \in V_1 ,  \quad& w_1 := u_1\times v_1 \in V_1\\
v_2 := u_2 \times C^Tu_1 \in V_2, \quad& w_2 := u_2\times v_2 \in V_2,
\end{align*}
and let $X_0 \subseteq X$ be the Zariski-open (dense) subset defined by $u_i\cdot u_i \neq 0$ and $v_i \cdot v_i \neq 0$, for $i=1,2$. Let $x = (u_1, u_2, C) \in X_0$ and let $\vec{u_i} = u_i/\lVert u_i \rVert$, $\vec{v_i} = u_i/\lVert v_i \rVert, \vec{w_i}= w_i/\lVert w_i \rVert, i=1,2$ be the corresponding normalized unit vectors. As explained in Section \ref{sec:symStates}, we may choose $g_1\in \SO(V_1) = \SO_3(\mathbb{C})$ so that $e_1 = g_1\vec{u_1}, e_2 = g_1\vec{v_1}, e_3 = g_1\vec{w_1}$, and $g_2\in \SO(V_2) = \SO_3(\mathbb{C})$ so that $e_1 = g_2\vec{u_2}, e_2 = g_2\vec{v_2}, e_3 = g_2\vec{w_2}$. We then have 

\begin{align}
\label{eqn:generalSection}
g_1 C g_2^t &= \left(\begin{array}{ccc}
\vec{u_1}\cdot C\vec{u_2} & 0 & \vec{u_1}\cdot C\vec{w_2} \\ 
0 & \vec{v_1}\cdot C\vec{v_2} & \vec{v_1}\cdot C\vec{w_2} \\ 
\vec{w_1}\cdot C\vec{u_2} & \vec{w_1}\cdot C\vec{v_2} & \vec{w_1}\cdot C\vec{w_2}
\end{array} \right), \\
\quad g_1u_1 &= \left(\begin{array}{c}
\lVert u_1 \rVert \\ 
0 \\ 
0
\end{array} \right),\quad
\quad g_2u_2 = \left(\begin{array}{c}
\lVert u_2 \rVert \\ 
0 \\ 
0
\end{array} \right),
\end{align} 
so that $GS$ is Zariski-dense in $X$. The calculation of the Weyl group $W(S)$ is similar to the case of symmetric states. In particular, note first that if $g = (g_1,g_2) \in G$ preserves the subvariety $S$, then each $g_i$ must preserve the subspace of $V_i$ spanned by $e_1$, so that $g_1,g_2$ are of the form 

$$
g_1 = \left( \begin{tabular}{c|c}
 $\pm 1$ & $0$ \\
 \hline
 $0$ & $h_1$  \\ 
 \end{tabular}\right), \quad 
 g_2 = \left( \begin{tabular}{c|c}
 $\pm 1$ & $0$ \\
 \hline
 $0$ & $h_2$  \\ 
 \end{tabular}\right), 
$$
with $h_1,h_2$ orthogonal 2-by-2 matrices. Given any element $x = (u_1,u_2,C) \in S$, write  $C = \left( \begin{tabular}{c|c}
 $a$ & $b_2^t$ \\
 \hline
 $b_1$ & $C_0$  \\ 
 \end{tabular}\right)
$ in block form, where $b_i^t = (0, b_{i0}), i=1,2$. Then $g$ belongs to $N(S)$ if and only if each $h_i$ preserves the subspace spanned by each $b_i$, so that both $g_1, g_2$ are of the form 
$$
g_1,g_2 =  \left(\begin{array}{cc}
\pm 1 & 0 \\ 
0 & \pm 1
\end{array} \right).
$$
It now follows that the Weyl group
$$
W(S) \simeq K_1\times K_2
$$ 
is a finite abelian group of order 16, abstractly isomorphic to two copies $K_1, K_2$ of $\mathbb{Z}/2\mathbb{Z}\times \mathbb{Z}/2\mathbb{Z}$, the Klein 4-group. This calculation also shows that $S$ is a relative section, since if two elements $x_1, x_2 \in S$ are in the same $G$-orbit, then they must be in the same $W(S)$-orbit. 

As in the case of symmetric states, we may use the relative section $S$ to give a complete description of the field of invariants $\mathbb{C}(X)^G$. First, since the field of invariants of any finite-dimensional complex representation of a finite abelian group is rational, again by Fischer's Theorem.  Applying \eqref{eqn:relSecIso}, we deduce the following:

\begin{thm}
\label{thm:rationalityGeneralCase}
The field of local unitary invariants $\mathbb{C}(\mathscr{L})^{\SO_3(\mathbb{C})\times \SO_3(\mathbb{C}) }$ of mixed states of two qubits is rational (i.e. purely transcendental) of degree 9. 
\end{thm}

Second, we can use $S$ and the orbit separation criterion (Thm. \ref{thm:separationTest}) to give an explicit, simple construction of nine (algebraically independent) generators $f_1, \ldots, f_9$ of $\mathbb{C}(X)^G$. In particular, let 
\begin{align}
\label{eqn:generalInvariants}
f_1 &= u_1\cdot u_1,\quad  f_2 = u_2\cdot u_2,\quad  f_3 = v_1\cdot v_1, \quad f_4 = v_2\cdot v_2, \\
f_5 &= u_1\cdot Cu_2,\quad f_6 = v_1\cdot Cv_2,\quad f_7 = w_1\cdot Cw_2,  \\
f_8 &= v_1\cdot Cw_2, \quad f_9 = w_1\cdot Cv_2.
\end{align} 

Then we have:

\begin{thm}
\label{thm:generationGeneralCase}
The nine invariants $f_1, \ldots f_9$ are algebraically independent and they generate the field of local unitary invariants $\mathbb{C}(\mathscr{L})^{\SO_3(\mathbb{C})\times \SO_3(\mathbb{C}) }$ of mixed states of two qubits.
\end{thm}

\begin{proof}
The proof is analogous to that of Thm. \ref{thm:symmetricGeneration}, for the case of symmetric states. Note first  that the action of $g = (g_1,g_2) \in G $ sends $u_i \mapsto g_iu_i$ and $C\mapsto g_1Cg_2^t$, so that $v_i$ is sent to $g_iv_i$ and $w_i = u_i\times v_i$ is sent to $g_iw_i$. From this it follows easily that all the nine functions $f_1, \ldots, f_9$ are $G$-invariant. Next, we show that these nine functions separate orbits in general position (Thm. \ref{thm:separationTest}). Let $X_0\subseteq X$ be the $G$-invariant, Zariski-open set given by $u_i\cdot u_i \neq 0$ and $v_i\cdot v_i \neq 0, i=1,2$. Suppose        
$$
x = \left( \begin{tabular}{c|c}
 1 & $u_2^t$ \\
 \hline
 $u_1$ & $C$ \\ 
 \end{tabular} \right),\quad 
 x' = \left( \begin{tabular}{c|c}
 1 & $(u_2')^t$ \\
 \hline
 $u_1'$ & $C'$ \\ 
 \end{tabular} \right),
$$
are two elements of $X_0$ satisfying $f_i(x) = f_i(x')$ for $i=1,\ldots, 9$. By the same argument used earlier to prove that $GS$ is Zariski-dense , there exist an element $g = (g_1,g_2)\in G$ so that $gx$ is in the form \eqref{eqn:generalSection}. Using the same identities for cross-products as in Thm. \ref{thm:symmetricGeneration}, we can write 
\begin{align*}
w_1\cdot Cu_2 &= (u_1\cdot Cu_2)^2 - (u_1\cdot u_1)(Cu_2\cdot Cu_2) = -v_1\cdot v_1 = -f_3 \\
u_1\cdot Cw_2 &= C^tu_1\cdot w_2 = (u_2\cdot C^tu_1)^2 - (u_2\cdot u_2)(C^tu_1\cdot C^tu_1) = -v_2\cdot v_2 = -f_4,
\end{align*}
so that we may express $gx$ in terms of the invariant functions $f_1(x), \ldots, f_9(x)$ as follows:

\begin{align*}
g_1 C g_2^t &= \left(\begin{array}{ccc}
f_5/(f_1f_2)^{1/2}& 0 & -f_4^{1/2}/(f_1f_2)^{1/2} \\ 
0 & f_6/(f_3f_4)^{1/2} &  f_8/(f_2f_3f_4)^{1/2} \\ 
-f_3^{1/2}/(f_1f_2f_3)^{1/2} & f_9/(f_1f_3f_4)^{1/2} & f_7/(f_1f_2f_3f_4)^{1/2}
\end{array} \right), \\
\quad g_1u_1 &= \left(\begin{array}{c}
f_1^{1/2} \\ 
0 \\ 
0
\end{array} \right), \quad g_2u_2 = \left(\begin{array}{c}
f_2^{1/2} \\ 
0 \\ 
0
\end{array} \right).
\end{align*} 
Similarly, we can find $g' = (g_1', g_2')\in G$ so that $g'x'$ is an identical expression in terms of the functions $f_1(x'), \ldots, f_9(x')$. Because $f_i(x) = f_i(x')$ for all $i$, it follows that $gx = g'x'$ and therefore $x = g^{-1}g'x'$, so that $x,x'$ are in the same $G$-orbit. By Thm. \ref{thm:separationTest}, the set $f_1, \ldots, f_9$ must generate the field $\mathbb{C}(X)^{G}$, which is rational of degree nine, and so  $\{f_1, \ldots, f_9\}$ must be algebraically independent.  
\end{proof}

As is the case for symmetric states, the generators $f_1, \ldots, f_9$ given are polynomial invariants, and their precise expression in terms of the coordinates of $x = (u_1, u_2, C)$ can easily be found using well-known formulas for cross-products and dot-products.

\section{Real-valued invariants}
\label{sec:realValued}

For applications, it is also important to find generators for the field of real-valued rational invariants $\mathbb{R}(\mathscr{L})^{\SU(2)\times\SU(2)}$, and not just for its complexification $\mathbb{C}(\mathscr{L}_{\mathbb{C}})^{\SO_3(\mathbb{C})\times \SO_3(\mathbb{C}) }$. Because the generating polynomial invariants $f_1, \ldots, f_9$ given in \eqref{eqn:generalInvariants} all have coefficients in $\mathbb{R}$ (in fact, they have {\em integer} coefficients), they all restrict to $\mathscr{L}$ to give real-valued invariants $\mathbb{R}(\mathscr{L})^{\SU(2)\times\SU(2)}$. We claim that these invariants generate  the whole field, so that in particular we also get rationality over $\mathbb{R}$:

\begin{thm}
\label{thm:generationRealInvariants}
The field of real-valued invariants $\mathbb{R}(\mathscr{L})^{\SU(2)\times\SU(2)}$ for mixed states of two qubits is rational, generated by the algebraically independent invariants $f_1, \ldots, f_9$ given in \eqref{eqn:generalInvariants}. 
\end{thm}

\begin{proof}

We offer a geometric proof based on the scheme-theoretic methods of \cite{CT-S} and \cite{GIT}. For ease of notation, let $X = \mathscr{L}$ and let $G = \SU(2)\times \SU(2)$. By Rosenlicht's theorem (see e.g. \cite[I.2]{GIT} or \cite[2.2]{CT-S}), there exists a Zariski-open (dense) open subset $U\subseteq X$ such that $Y = U/G$ is a geometric quotient, so that, in particular, the complex points $Y(\mathbb{C})$ correspond to the $G$-orbits on complexified Liouville space $\mathscr{L}_{\mathbb{C}}$, and the function field $\mathbb{R}(Y)$ is equal to the field of invariants $\mathbb{R}(X)^G$. On the other hand, let $Z = \Spec(\mathbb{R}[f_1, \ldots, f_9]) \simeq \mathbb{A}_{\mathbb{R}}^9$ be the irreducible affine variety whose coordinate ring is generated by the (algebraically independent) invariants $f_1, \ldots, f_9$. The inclusion of fields $\mathbb{R}(f_1, \ldots, f_9) \subseteq \mathbb{R}(X)^G$ gives a rational, dominant map $\phi: Y \rightarrow Z$. Because the invariants $f_1, \ldots, f_9$ separate orbits in general position over $\mathbb{C}$ (as shown in the proof of Thm. \ref{thm:generationGeneralCase}), it follows that $\phi$ is radicial (i.e. universally injective). By \cite[I.3.5.8]{EGA1}, this implies that $\mathbb{R}(X)^G$ is a radicial (i.e. purely inseparable) extension of $\mathbb{R}(f_1, \ldots, f_9)$. But in characteristic zero, there are no non-trivial radicial extensions, and therefore $\mathbb{R}(f_1, \ldots, f_9) = \mathbb{R}(X)^G$.  

\end{proof}

Note that a similar statement can be obtained for the case of symmetric states. In particular, the same argument shows that the field of real-valued invariants $\mathbb{R}(\mathscr{L}_{\rm{sym}})^{\SU(2)}$ is rational, generated by the six invariants given in Thm. \ref{thm:symmetricGeneration}.

 \section*{Acknowledgments} 
The authors are supported by the U.S. Department of Energy, Office of Science, Basic Energy Sciences, under Award Number DE-SC-SC0022134.
  The last author is also partially supported by an OVPR Postdoctoral Award at Wayne State University.

\bibliography{entanglement}

\end{document}